\def \cls {.} %path to common latex files (change for your own relative/absolute path)
\def \pics {.}      %path to pics files (change for your own relative/absolute path)

\documentclass[letter, 10pt, conference]{ieeeconf}
\IEEEoverridecommandlockouts
\usepackage[ieeeconf,noglossaries]{\cls/standard}

\usepackage{etex}

\DeclareMathOperator{\chull}{Convex\,Hull}
\newcommand{\pref}[1]{Property~\ref{#1}}

\title{\LARGE \bf Obstacle avoidance via B-spline parameterizations of flat trajectories}

\author{Florin Stoican$^\dagger$\thanks{$^\dagger$ Politehnica University of Bucharest, Department of Automation Control and Systems Engineering, Bucharest, Romania {\tt\small florin.stoican@acse.pub.ro, vlad\_mihai.ivanusca@aut.pub.ro}}, Vlad-Mihai Iv\u anu\c sc\u a$^\dagger$, Ionela Prodan$^\ddagger$\thanks{$^\ddagger$, Univ. Grenoble Alpes, Laboratory of Conception and Integration of Systems (LCIS EA 3747), F-26902, Valence, France, {\tt\small ionela.prodan@lcis.grenoble-inp.fr}}}%, Dan Popescu$^*$\thanks{$^*$ Department of Automatic Control and Industrial Informatics, UPB, Romania, {\tt\small dan.popescu@upb.ro}}}
\begin{document}

\maketitle

\begin{abstract}
This paper considers the collision avoidance problem in a multi-agent multi-obstacle framework. The originality in solving this intensively studied problem resides in the proposed geometrical view combined with differential flatness for trajectory generation and B-splines for the flat output parametrization. Using some important properties of these theoretical tools we show that the constraints can be validated at all times. Exact and sub-optimal constructions of the collision avoidance optimization problem are provided. The results are validated through extensive simulations over standard autonomous aerial vehicle dynamics.
\end{abstract}

%%%%%%%%%%%%%%%%%%%%%%%%%%%%%%%%%%%%%%%%%%%%%%%%%%%%%%%%%%%%%%%%%%%%%%%%%%%%%%%%
\section{Introduction}

One of the main issues in multi-obstacle, multi-agent environments is the collision avoidance assessment. Usually the avoidance constraints have to be considered both between an agent and fixed obstacles and between any two agents. The problem is intensively studied, yet actual, in the literature but is usually tilted towards heuristic approaches or online validations. The first usually lacks stability and performances guarantees and the later is easily boggled into numerical issues \cite{burger2010smooth, aguiar2007trajectory}.

An alternative approach is to solve offline the difficult trajectory generation part of the overall problem, online only a straightforward trajectory tracking is employed. This reduces significantly the online computations and allows stability and performance analysis. The caveat is that the trajectory to be computed has to respect the dynamics of the agent and to validate the collision avoidance constraints at all times \cite{mellinger2012trajectory,prodan2013receding}.

An interesting implementation is represented by flat trajectory design which guarantees that the corresponding system dynamics are respected (with the caveat that state and input constraints are not easily accounted for \cite{de2009flatness, suryawan2010methods}). This construction shifts the state and input constraints into constraints over the flat output. To handle this, usually, the flat output is project over some basis functions which means that only the coefficients of the projections need to be found. In this sense, B-spline functions represent an ideal choice since they have enough flexibility \cite{suryawan2012constrained} and nice theoretical properties (of which we will make extensive use throughout the paper).

The present work builds on results sketched in \cite{med2015} and further advances the topic in several directions. Foremost, we provide exact and sub-optimal formulations of the collision avoidance problems between an agent and the obstacles and between any tho agents. In both cases we make use of the geometrical properties of the B-spline functions which allow to bound locally the trajectories obtained through them. Hence, the collision problems become separation problems between sets of consecutive points. In the exact case, these constraints lead to nonlinear formulation where both the control points and the separation hyperplanes are variables. A simplified (and hence sub-optimal approach) is to select the separation hyperplanes from the support hyperplanes of the obstacles therefore reducing the problem to a mixed integer formulation.

Both methods make use of a multi-obstacole framework and are tested and compared over extensive simulations. The rest of the paper is organized as follows: \secref{sec:pre} tackles flat output and B-spline characterizations, \secref{sec:flat} discusses flat trajectory generation and \secref{sec:main} presents the main results of the paper which are then illustrated in \secref{sec:ill}. \secref{sec:con} draws the conclusions.

\subsection*{Notation}
The Minkowski sum of two sets, $A$ and $B$ is denoted as $ A\oplus B=\left\{x:\: x=a+b,\: a\in A,b\in B\right\}$. $\chull\{p_1\dots p_n\}$ denotes the convex hull of set generated by the collection of points $p_1\dots p_n$.

\section{Preliminaries}
\label{sec:pre}

The problem of designing reference trajectories in a multi-agent multi-obstacle environment is in general a difficult one. A popular approach is to parametrize them through flatness constructions \cite{fliess1995flatness,levine2009analysis,de2009flatness}. In this section we will describe some of the basics of flat trajectory and their parametrization via B-spline basis functions.  

\subsection{Flat trajectories}

A nonlinear time invariant system: 
\be
\label{eq:contSys}
\dot{x}(t)=f(x(t),u(t)),
\ee
where $x(t)\in \re{n}$ is the state vector and $u(t) \in \re{m}$ is the input vector is called differentially flat if there exists the flat output $z(t) \in \re{m}$:
\be
\label{eq:flat_out}
z(t)=\gamma(x(t),u(t),\dot{u}(t),\cdots, u^{(q)}(t))
\ee
such that the states and inputs can be algebraically expressed in terms of $z(t)$ and a finite number of its higher-order derivatives:
\bea
\label{eq:diff_a}
x(t)&=&\Theta(z(t),\dot{z}(t), \cdots ,z^{(q)}(t)),\\ \nonumber
u(t)&=&\Phi(z(t),\dot{z}(t),\cdots ,z^{(q)}(t)).
\eea
\begin{rem}
\label{rem:flatInput}
For any system admitting a flat description, the number of flat outputs equals the number of inputs \cite{levine2009analysis}. In the case of linear systems \cite{sira2004differential} the flat differentiability (existence and constructive forms) is implied by the controllability property.\eot
\end{rem}

Within the multi-agent framework, the most important aspect of construction \eqref{eq:flat_out}--\eqref{eq:diff_a} is that it reduces the problem of trajectory generation to finding an adequate flat output \eqref{eq:flat_out}. This means choosing $z(t)$ such that, via mappings $\Theta(\cdot),\Phi(\cdot)$, various constraints on state and inputs \eqref{eq:diff_a} are verified. Since the flat output may be difficult to compute under these restrictions, we parametrize $z(t)$ using a set of smooth basis functions $\Lambda^i(t)$:
\be
\label{eq:paramFlat}
z(t)=\sum\limits_{i=1}^{N}{\alpha_i \Lambda^i(t)}, \ \  \alpha_i \in \re{}.
\ee 
Parameter $N$ depends on the number of constraints imposed onto the dynamics \cite{wilkinson}.

There are multiple choices for the basis functions $\Lambda^i(t)$. Among these, \emph{B-spline} basis functions are well-suited to flatness parametrization due to their ease of enforcing continuity and because their degree depends only up to which derivative is needed to ensure continuity \cite{suryawan2012constrained,de2009flatness}.

\subsection{B-splines}
\label{sec:bsplines}

A B-spline of order $d$ is characterized by a \emph{knot-vector} \cite{gordon1974b,patrikalakis2009shape} 
\be
\label{eq:knot_vector}
\mathbb T=\left\{\tau_0,\tau_1\dots  \tau_m\right\},
\ee
of non-decreasing time instants ($\tau_0\leq \tau_1\leq\dots \leq \tau_m$) which parametrizes the associated basis functions $B_{i,d}(t)$:
\bse
\begin{align}
\label{eq:bsplines_0}
B_{i,1}(t)&=\begin{cases}1, \textrm{ for }\tau_i\leq t<\tau_{i+1}\\0\textrm{ otherwise}\end{cases},\\
\label{eq:bsplines_d}
B_{i,d}(t)&=\frac{t-\tau_i}{\tau_{i+d-1}-\tau_i}B_{i,d-1}(t)+\frac{\tau_{i+d}-t}{\tau_{i+d}-\tau_{i+1}}B_{i+1,d-1}(t)
\end{align}
\ese
for $d>1$ and $i=0,1\dots n=m-d$.
%\begin{rem}
%Note that the B-splines \eqref{eq:bsplines_0}--\eqref{eq:bsplines_d} partition the unity, in the sense that i) $B_{i,d}(t)\geq 0$ and ii) $\sum\limits_{i=0}^n B_{i,d}(t)=1$ for all $t\in [t_0,t_m]$.\eor 
%\end{rem}

Considering a collection of \emph{control points} 
\be
\label{eq:control_points}
\mathbb P=\left\{p_0,p_1\dots  p_n\right\},
\ee
we define a \emph{B-spline curve} as a linear combination of the control points \eqref{eq:control_points} and the B-spline functions \eqref{eq:bsplines_0}--\eqref{eq:bsplines_d} 
\be
\label{eq:bspline_curve}
z(t)=\sum\limits_{i=0}^n B_{i,d}(t)p_i=\mathbf P\mathbf B_d(t)
\ee
where $\mathbf P=\bbm p_0\dots  p_n\ebm$ and $\mathbf B_d(t)=\bbm B_{0,d}(t)\dots  B_{n,d}(t)\ebm^T$.
This construction yields several interesting properties \cite{piegl1995curve}:
\begin{enumerate}
\item[P1)]\label{p:1} $z(t)$ is $C^\infty$ in any $t\notin \mathbb T$ and $C^{d-1}$ in any $t\in \mathbb T$;
\item[P2)]\label{p:3} at a time instant $\tau_i<t<\tau_{i+1}$, $z(t)$ depends only on the B-splines $B_{i-d+1,d}(t)\dots  B_{i,d}(t)$; consequently, the B-spline curve $z(t)$ lies within the union of all convex hulls formed by all $d$ successive control points;
\item[P3)]\label{p:5} the `r' order derivatives of B-spline basis functions can be expressed as linear combinations of B-splines of lower order ($\mathbf B_{d}^{(r)}(t)=M_r \mathbf B_{d-r}(t)$ with matrices $M_r$ of appropriate dimensions and content);
\item[P4)] taking the first and last $d$ knot elements equal ($\tau_0=\dots = \tau_{d-1}$ and $\tau_{n+1}=\dots \tau_{n+d}$) leads to a \emph{clamped B-spline curve} where the first and last control points coincide with the curve's end points.
\end{enumerate}
%\begin{center}
%\begin{figure}[!ht]
%\subfloat[B-splines basis functions]{\label{fig:basis}\includegraphics[width=.5\columnwidth]{\pics/bsplines/bspline_basis}}\hfill
%\subfloat[B-splines basis derivatives]{\label{fig:basis_deriv}\includegraphics[width=.5\columnwidth]{\pics/bsplines/bspline_derivatives}}
%\caption{B-spline basis functions and their first order derivatives.}
%\label{fig:bsplines}
%\end{figure}
%\end{center}

\subsection*{Illustrative example}

Let us consider for exemplification the B-spline basis functions defined by parameters $n=5$ and $d=4$ with the knot vector of length $m=n+d=9$ and with components equally sampled between $0$ and $1$. Using these elements we construct the B-spline curve \eqref{eq:bspline_curve} which we depict in \figref{fig:bsplines_curve}.
\begin{figure}[!ht]
\centering
\includegraphics[width=.9\columnwidth]{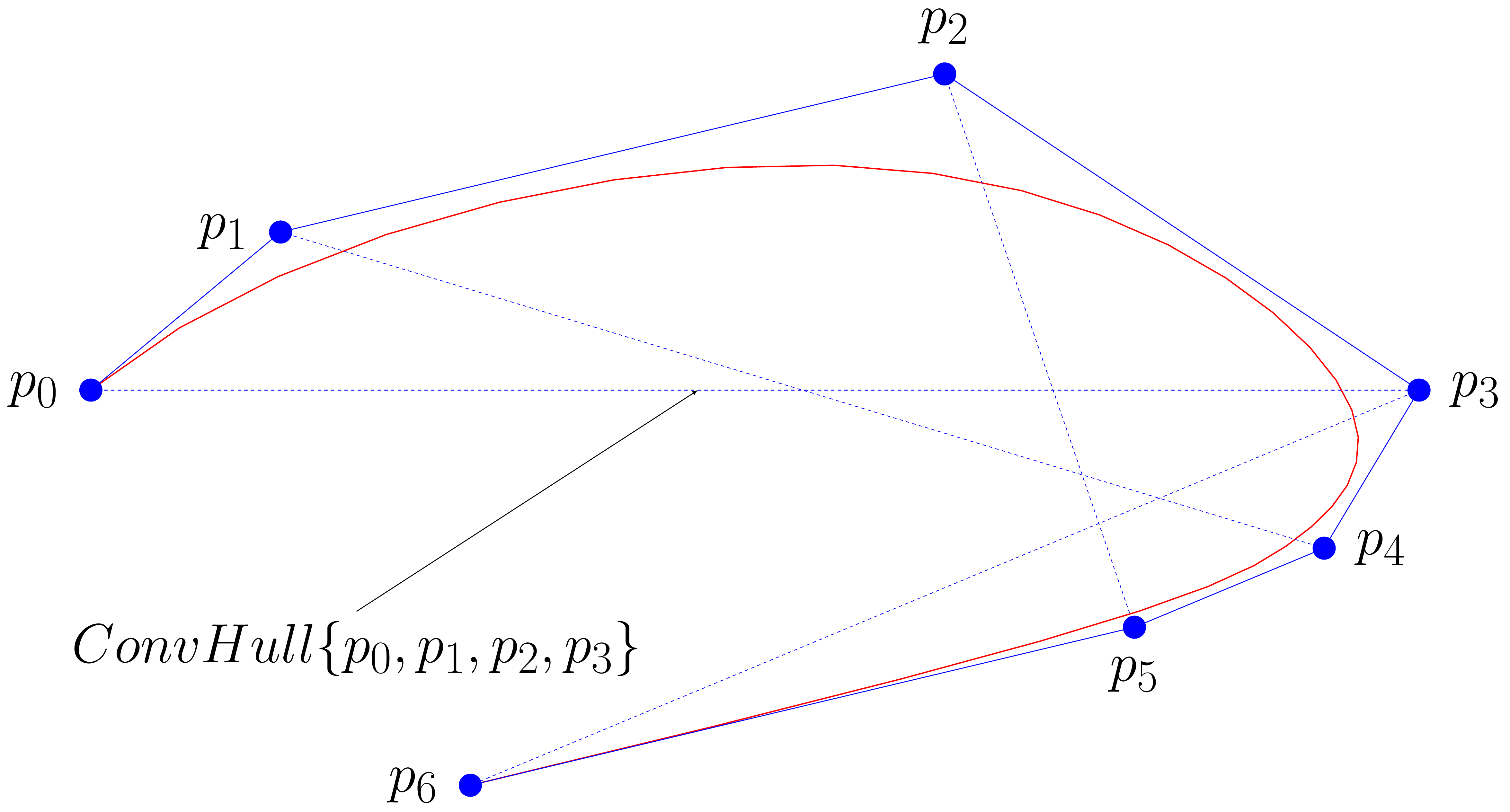}
\caption{B-spline curve (solid red) and its control polygone (dashed blue).}
\label{fig:bsplines_curve}
\end{figure}
Note that the convex-hulls determined by all $d=4$ consecutive control points (\pref{p:3}) contain the curve:
\ben
z(t)\subset \bigcup\limits_{i=3\dots 5}{\chull}\{p_{i-3},p_{i-2},p_{i-1},p_{i}\},\: \forall t\in [0,1].
\een
In particular, the patterned shape depicted in the figure denotes region ${\chull}\{p_{2},p_{3},p_{4},p_{5}\}$ which contains the curve $z(t)$ for any $t\in [\tau_2,\tau_3]$.

\section{Flat trajectory generation}
\label{sec:flat}

Let us consider a collection of $N+1$ way-points and the time stamps associated to them\footnote{In particular we may consider just two way-points: the start and end points of the trajectory.}:
\be
\label{eq:wt}
\mathbb W=\{w_s\}\textrm{ and } \mathbb T_{\mathbb W}=\{t_s\},
\ee
for any $s=0\dots N$. The goal is to construct a flat trajectory which passes through each way-point $w_s$ at the time instant $t_s$ (or through a predefined neighborhood of it \cite{med2015}), i.e., to find a flat output $z(t)$ such that 
\be
\label{eq:xconstr}
x(t_s)=\Theta(z(t_s),\dots z^{(r)}(t_s))=w_s, \: \forall s=0\dots  N.
\ee 
\begin{rem}
\label{rem:waypoints}
Note that here we assume that the way-points are defined over the entire state. Arguably there might be situations where only a subspace of the state is of interest (e.g., only the position components of the state).\eor
\end{rem}
Making use of the B-spline framework we provide a vector of control points \eqref{eq:control_points} and its associated knot-vector \eqref{eq:knot_vector} such that \eqref{eq:xconstr} is verified (parameter $d$ is chosen such that continuity constraints are respected):
\be
\label{eq:flatconstraints}
\tilde\Theta(\mathbf B_d(t_s),\mathbf P)=w_s,\: \forall s=0\dots  N,
\ee
where $\tilde \Theta(\mathbf B_d(t),\mathbf P)=\Theta (\mathbf P\mathbf B_d(t)\dots  \mathbf P M_r\mathbf B_{d-r}(t))$ is constructed along \pref{p:5}. 

Let us assume that the knot-vector is fixed ($\tau_0=t_0$, $\tau_{n+d}=t_N$ and the intermediary points $\tau_{1}\dots  \tau_{n+d-1}$ are equally distributed along  these extremes). Then, we can write an optimization problem with control points $p_i$ as decision variables whose goal is to minimize a cost $\Xi(x(t),u(t))$ along the time interval $[t_0,t_N]$:
\be
\label{eq:flatcost}
\begin{split}
\mathbf P=&\arg\min\limits_{\mathbf P} \int_{t_0}^{t_N}||\tilde \Xi(\mathbf B_d(t),\mathbf P)||_Qdt\\
&\textrm{s.t. constraints \eqref{eq:flatconstraints} are verified}
\end{split}
\ee
with $Q$ a positive symmetric matrix. The cost \[\tilde \Xi(\mathbf B_d(t),\mathbf P)=\Xi(\tilde \Theta(\mathbf B_d(t),\mathbf P),\tilde \Phi(\mathbf B_d(t),\mathbf P))\] can impose any penalization we deem necessary (length of the trajectory, input variation, input magnitude, etc). In general, such a problem is nonlinear (due to mappings $\tilde\Theta(\cdot)$ and $\tilde\Phi(\cdot)$) and hence difficult to solve. A nonlinear MPC iterative approach has been extensively studied \cite{de2009flatness}. 

With these tools at hand we can propose various methods for collision avoidance in a multi-agent multi-obstacle environment.
\section{Main idea}
\label{sec:main}
Let us consider a collection of polyhedral obstacles
\be
\label{eq:obstacles}
\mathbb O=\{O_1\dots O_{N_o}\}
\ee
and assume that the k-th agent follows a trajectory $r_k(t)$ during the interval $[t_0, t_N]$, generated as in \eqref{eq:xconstr} through a collection of control points\footnote{For convenience we keep a common degree `d' and number of control points `n'. This is a reasonable assumption as long as the agents have a common dynamic and follow similar restrictions.} $\mathbb P^k=\{p_j^k\}$ and the associated knot vector $\mathbb T^k=\{\tau_j^k\}$. 

Consequently, the collision avoidance conditions mentioned earlier can be formulated as follows:
\begin{enumerate}
\item[(i)] collision avoidance between the k-th agent and l-th obstacle: 
\be
\label{eq:real_avoidance_ao}
r_k(t)\notin O_l, \forall t\in [t_0,t_N],
\ee
\item[(ii)] collision avoidance between the $k_1$-th and $k_2$-th agents (for any $k_1\neq k_2)$:
\be
\label{eq:real_avoidance_aa}
r_{k_1}(t)\neq r_{k_2}(t), \forall t\in [t_0,t_N].
\ee
\end{enumerate}

\subsection{The exact case}

The distinctive feature of conditions \eqref{eq:real_avoidance_ao}--\eqref{eq:real_avoidance_aa} is that they require a continuous time interval ($[t_0,t_N]$) validation (i.e., imposing constraints at discrete time instants $t_k$ along the interval is not deemed sufficient). Consequently, we make use of property \pref{p:3} which allows to bound the continuous B-spline parametrized curve by its control points. Coupling this with the \emph{separating hyperplane theorem} (a well-known construction \cite{sion1958general} which states that  for any two disjoint convex objects there exists a separating hyperplane) several results are attainable.

First, we provide a slight reformulation of Proposition 1 from \cite{stoican2015flat}.
\begin{prop}
\label{prop:avoid}
The k-th agent is guaranteed to avoid obstacles \eqref{eq:obstacles}, i.e., to verify \eqref{eq:real_avoidance_ao}, if there $\exists c_{il}^k\in \re{n}$ s.t.
\be
\label{eq:avoidance}
\max\limits_{j\in\{i-d+1\dots  i\}}\left(c_{il}^{k}\right)^\top p_j^k\leq \min\limits_{x\in \tilde\Theta^{-1}(O_l)} \left(c_{il}^k\right)^\top x,
\ee
for $i=d-1\dots  n$ and $\forall O_l\in \mathbb O$.\eot
\end{prop} 
\begin{proof}
Condition \eqref{eq:avoidance} states that there exists a hyperplane defined by its normal $c_{il}^{k}$ which separates the points $\{p_{i-d+1}^k\dots p_{i}^k\}$ from the obstacle $\tilde\Theta^{-1}(O_l)$. Since, according to \pref{p:3}, the curve \eqref{eq:xconstr} is contained in $\cup_{i=d\dots n}\chull\{p_{i-d+1}^k\dots p_{i}^k\}$ it follows that \eqref{eq:avoidance} is a sufficient condition to verify \eqref{eq:real_avoidance_ao}.\qed
\end{proof}
\begin{rem}
In \propref{prop:avoid} note the use of mapping $\tilde\Theta^{-1}(\cdot)$. This appears because the obstacle avoidance constraint is the state-space whereas \eqref{eq:avoidance} is in the control point space.\eor
\end{rem}

A similar reasoning is employed for the inter-agent collision condition \eqref{eq:real_avoidance_aa}.
\begin{prop}
\label{prop:avoid_inter}
The pair ($k_1,k_2$) of agents, with $k_1\neq k_2$, is guaranteed to avoid collision, i.e., to validate \eqref{eq:real_avoidance_aa}, if there $\exists c_{i_1i_2}^{k_1k_2}\in \re{n}$ s.t.:
\be
\label{eq:avoid_inter}
\max\limits_{j\in \{i_1-d+1\dots  i_1\}}{\left(c_{i_1i_2}^{k_1k_2}\right)^\top p_j^{k_1}}\le \!\!\!\!\!\!\!\!\min\limits_{j\in \{i_2-d+1\dots  i_2\}}{\left(c_{i_1i_2}^{k_1k_2}\right)^\top p_j^{k_2}},
\ee
for all possible pairs $(i_1, i_2)$ which validate
\be
\label{eq:intersection_aa_indices}
\{(i_1,i_2):\: [\tau_{i_1}^{k_1},\tau_{i_1+1}^{k_1}]\cap [\tau_{i_2}^{k_2},\tau_{i_2+1}^{k_2}]\neq \emptyset\}.
\ee \eot
\end{prop}
\begin{proof}
Recall that (as per \pref{p:3}) a region $\chull\{p_{i-d+1}\dots p_{i}\}$ contains the B-spline curve in the time interval $[\tau_{i}, \tau_{i+1}]$. Applying this to the agents $k_1$ and $k_2$ means that all regions corresponding to indices \eqref{eq:intersection_aa_indices} should not intersect as they contain overlapping time instants. The separation is enforced by \eqref{eq:avoid_inter} which is a sufficient condition for \eqref{eq:real_avoidance_aa}.
\qed
\end{proof}

\begin{rem}
\label{rem:avoid_simpl}
Eq. \eqref{eq:intersection_aa_indices} can be avoided altogether if the B-spline parametrizations share the same knot vector (i.e., $\mathbb T^{k_1}=\mathbb T^{k_2}$). In such a case, variable $c_{i_1i_2}^{k_1k_2}$ becomes $c_{i}^{k_1k_2}$ and condition \eqref{eq:avoid_inter} is simplified to 
\be
\label{eq:avoid_inter2}
\max\limits_{j\in \{i-d+1\dots  i\}}{\left(c_{i}^{k_1k_2}\right)^\top p_j^{k_1}}\le \min\limits_{j\in \{i-d+1\dots i\}}{\left(c_{i}^{k_1k_2}\right)^\top p_j^{k_2}},
\ee
for all $i\in \{d-1\dots  n\}$. \eor
\end{rem}

\subsection{The sub-optimal case}

Verifying \eqref{eq:avoidance} (or \eqref{eq:avoid_inter}) is difficult in practice due to the presence of bi-linear terms (e.g., in \propref{prop:avoid} both $c_{il}$ and $p_j^k$ are variables). Hereafter we propose a simpler (and hence sub-optimal) implementation. 

The main idea is that instead of letting the separating hyperplane from \propref{prop:avoid} or~\ref{prop:avoid_inter} be itself a variable, we choose from within a predefined pool of hyperplanes. A natural choice is to select from the support hyperplanes of the obstacles. By definition, such a hyperplane contains on one side the obstacle and hence, it remains only to check whether the control points lie on the opposite side. The selection of the active hyperplane is done through decision variables (i.e., binary variables) which leads to a mixed-integer pseudo-linear formulation. 

To generate the collection of hyperplanes, we consider the polyhedral sets bounding\footnote{We assume that $\tilde \Theta^{-1}(O_l)$ is a polyhedral set as well. In general this might not hold, but in that case a polyhedral approximation can be obtained.} $\tilde \Theta^{-1}(O_l)$ and take the support hyperplanes which characterize them:
\be
\label{eq:hyp}
\mathcal H_m=\{x:\:h_m^\top x=k_m\}, \forall m=1\dots M.
\ee
Each of these hyperplanes partitions the space in two ``half-spaces'':
\begin{align}
\label{eq:hyp_plus}
\mathcal H_m^+&=\{x:\:\hphantom{-}h_m^\top x\leq \hphantom{-}k_m\},\\
\label{eq:hyp_minus}
\mathcal H_m^-&=\{x:\:-h_m^\top x\leq -k_m\}.
\end{align}
Taking into account all possible combinations of half-spaces leads to a hyperplane arrangement which divides the space into a collection of disjoint cells which are completely characterized by sign tuples \cite{ferrez2001cuts}:
\be
\label{eq:hyparr}
\mathbb H=\bigcup\limits_{\sigma\in \Sigma} \mathcal A(\sigma)=\bigcup\limits_{\sigma\in \Sigma} \left(\bigcup\limits_{m=1}^M \mathcal H_m^{\sigma(m)}\right)
\ee
where $\Sigma\subset \{-,+\}^M$ denotes the collection of all feasible (corresponding to non-empty regions $\mathcal A(\sigma)$) sign tuples. Each of these tuples can be allocated to either 
\begin{enumerate}
\item the admissible domain $\re{n}\setminus \mathbb O$:
\be
\label{eq:sigma_adm}
\Sigma^\circ=\{\sigma:\: \mathcal A(\sigma)\cap \mathbb O =\emptyset\},
\ee
\item or the interdicted domain $\mathbb O$:
\be
\label{eq:sigma_int}
\Sigma^\bullet=\{\sigma:\: \mathcal A(\sigma)\cap \mathbb O \neq \emptyset\},
\ee
\end{enumerate}
where $\Sigma^\bullet\cap \Sigma^\circ=\emptyset$ and $\Sigma^\bullet\cup \Sigma^\circ=\Sigma$. With these elements we can provide the following corollaries.
\begin{cor}
\label{cor:avoid}
For an obstacle $\mathcal A(\sigma^\bullet)$ with $\sigma^\bullet\in \Sigma^\bullet$, a sufficient condition to guarantee\footnote{We make the simplifying assumption that to each obstacle $O_l$ corresponds a single sign tuple.} \eqref{eq:real_avoidance_ao} is:
\bse
\begin{align}
\label{eq:avoidance2a}-\sigma^\bullet(m)h_m^\top p_j^k&\leq -\sigma^\bullet(m)k_m+T\alpha_{im}^k,\nonumber\\
\qquad \quad &\forall m=1\dots M, j=i-d+1\dots i\\
\label{eq:avoidance2b}\sum\limits_{m=1}^{M}\alpha_{im}^k&\leq M-1
\end{align}
\ese
for $i\in\{d-1\dots n\}$.\eot
\end{cor}
\begin{proof}
Taking in \eqref{eq:avoidance2a} he binary variable `$\alpha_{im}^k=0$' means that the i-th region $\chull\{p_{i-d+1}^k\dots  p_i^k\}$ of the k-th agent sits on the opposite side of the obstacle $\mathcal A(\sigma^\bullet)$ with respect to the the m-th hyperplane. The converse, taking `$\alpha_{jm}^k=1$' means that inequality \eqref{eq:avoidance2a} is discarded since the right hand term is sufficiently large to ignore the values on the left side (assuming that `$T$' was taken as a sufficiently large positive constant).

Condition \eqref{eq:avoidance2b} forces that for any consecutive $d+1$ points at least one of the inequalities \eqref{eq:avoidance2a} is enforced since at least one of the variables $\alpha_{im}^k$ has to be zero.\qed
\end{proof}
%While choosing the hyperplanes from the support hyperplanes defining the obstacles is a natural choice it is not necessarily the best approach when validating \eqref{eq:real_avoidance_aa}, that is, any preset collection of hyperplanes would do. Nonetheless, in the interest of consistency, we keep the hyperplanes \eqref{eq:hyp} and
Furthermore, we relax \propref{prop:avoid_inter} (in addition to the simplification proposed in \remref{rem:avoid_simpl}) into the following corollary.
\begin{cor}
\label{cor:avoid_inter}
The pair ($k_1,k_2$) of agents, with $k_1\neq k_2$, is guaranteed to avoid collision, i.e., to validate \eqref{eq:real_avoidance_aa}, if:
\begin{align}
\label{eq:avoid_inter_sub}
\max\limits_{j\in \{i-d\dots  i\}}{h_m^\top p_j^{k_1}}&\le \min\limits_{j\in \{i-d+1\dots i\}}{h_m^\top p_j^{k_2}}+T\beta_{im}^{k_1k_2},\\
\label{eq:avoid_inter_sub2}
\sum\limits_{m=1}^{M}\beta_{im}^{k_1k_2}&\leq M-1
\end{align}
for $i\in\{d-1\dots n\}$.\eot
\end{cor}
\begin{proof}
The binary variables $\beta_{im}^{k_1k_2}$ denote whether the i-th regions $\chull\{p_{i-d+1}^{k_1}\dots p_i^{k_1}\}$ and $\chull\{p_{i-d+1}^{k_2}\dots p_i^{k_2}\}$ are separated through the m-th hyperplane (whenever $\beta_{im}^{k_1k_2}=0$ the inequality \eqref{eq:avoid_inter_sub} is enforced and otherwise is discarded). Eq. \eqref{eq:avoid_inter_sub2} assures that at least one of the hyperplanes is active.\qed
\end{proof}
Several remarks are in order.
\begin{rem}
\cororef{cor:avoid_inter} considers a simultaneous computation of trajectories. An alternative is to compute them iteratively such that from the point of view of the current agent the obstacles to be avoided at $t\in [\tau_i,\tau_{i+1}]$ become:
\be
\mathbb O\leftarrow \mathbb O\cup\left(\bigcup\limits_{l<k}\chull\{p_{i-d+1}^l\dots p_i^l\}\right)
\ee
where besides the obstacles \eqref{eq:obstacles}, the i-th regions $\chull\{p_{i-d+1}^l\dots p_i^l\}$ of the previous agents (with index $l<k$ and whose trajectories are hence already computed) are also considered as obstacles.  \eor
\end{rem}
\begin{rem}
An agent may have a safety region around it (i.e., because the agent cannot be reduced to a point or due to the presence of disturbances in the dynamics). Whatever the reason, and the modality to obtain it, a safety region $S_k$ can be attached to the k-th agent. Consequently, the collision avoidance constraints \eqref{eq:real_avoidance_ao}--\eqref{eq:real_avoidance_aa} become:
\be
\label{eq:real_avoidance_ao2}
\{r_k(t)\}\oplus S_k\notin O_l, \forall t\in [t_0,t_N],
\ee
and
\be
\label{eq:real_avoidance_aa2}
\{r_{k_1}(t)\}\oplus S_{k_1}\neq \{r_{k_2}(t)\}\oplus S_{k_2}, \forall t\in [t_0,t_N].
\ee
The previous results can be easily adapted to constraints \eqref{eq:real_avoidance_ao2}--\eqref{eq:real_avoidance_aa2} by enlarging the obstacles ($r_k(t)\notin O_l\oplus \{-S_k\}$) and by requiring a larger inter-distance between agents respectively ($r_{k_1}(t)-r_{k_2}(t)\notin \{-S_{k_1}\}\oplus S_{k_2}$). \eor
\end{rem}
\begin{rem}
Lastly, it is worth mentioning that in all previous propositions and corollaries there is no guarantee of feasibility for the optimization problems. The solution is to incrementally increase the number of variables (i.e., the control points) until a feasible solution is reached.\eor
\end{rem}

\section{Illustrative example for an UAV system}
\label{sec:ill}

We revisit the test case from \cite{med2015}. A 2D 3-DOF model \eqref{eq:modelN} of an airplane in which the autopilot forces coordinated turns (zero side-slip) at a fixed altitude:
\be
\label{eq:modelN}
\begin{split}
\dot{x}(t)&=V_\textrm{a}(t) \cos \Psi(t),\\
\dot{y}(t)&=V_\textrm{a}(t) \sin \Psi(t),\\
\dot{\Psi}(t)&=\frac{g \ \tan \Phi(t)}{V_\textrm{a}(t)}
\end{split}
\ee
The state variables are represented by the position $(x(t),y(t))$ and the heading (yaw) angle $\Psi(t) \in [0, 2 \pi]$ rad. The input signals are the airspeed velocity $V_\textrm{a}(t)$ and the bank (roll) angle $\Phi(t)$, respectively. 

We take as flat output the position components of the state, $z(t)=\bbm z_1(t)&z_2(t)\ebm^\top=\bbm x(t)&y(t)\ebm^\top$ which permits to compute the remaining variables:
\bse
\begin{align}
\label{eq:stateRef}
\Psi(t)&=\arctan \left(\frac{\dot{z}_\textrm{2}(t)}{\dot{z}_\textrm{1}(t)}\right), \\
V_a(t)&=\sqrt{\dot{z}^\textrm{2}_\textrm{1}(t) + \dot{z}^\textrm{2}_\textrm{2}(t)},\\
\Phi(t)&=\arctan\left(\frac{1}{g}\frac{\ddot{z}_\textrm{2}(t) \dot{z}_\textrm{1}(t) - \dot{z}_\textrm{2}(t)\ddot{z}_\textrm{1}(t)}{\sqrt{\dot{z}^\textrm{2}_\textrm{1}(t) + \dot{z}^\textrm{2}_\textrm{2}(t)}}\right).
\label{eq:inputRef}
\end{align}
\ese
Note that in the heading component of the state appear 1st order and in the roll angle input appear 2nd order derivatives of the flat outputs. Hence, if we wish to have smooth state and input (their derivatives to be continuous) it follows that the B-spline parametrization has to have at least degree $d=4$.

Further, we consider way-points which fix only the position components of the state and time-stamps at which the trajectory has to pass through them. Thus we manage to skirt some of the thornier numerical aspects: the dependence between the B-spline basis functions and the position components is linear ($\tilde \Theta(\mathbf B_d(t),\mathbf P)=\mathbf P\mathbf B_d(t)$), and hence the cost and constraints will be easily written.

We take as cost to be minimized the length of the curve since we would like to have the shortest path which respects the constraints, i.e., $\tilde \Xi(\mathbf B_{d}(t),\mathbf P)=||z'(t)||$. This translates into the integral cost:
\be
\begin{split}
&\int\limits_{t_0}^{t_N} ||z'(t)||dt=\int\limits_{t_0}^{t_N} ||\mathbf P M_1\mathbf B_{d-1}(t)||dt\\
&=\sum\limits_{i,j}([\mathbf P M_1]_i)^T\left(\int\limits_{t_0}^{t_N}B_{i,d-1}(t)B_{j,d-1}(t)dt\right)[\mathbf P M_1]_j
\end{split}
\ee
where matrix $M_1$ links $\mathbf B_{d-1}(t)$ and $\dot{\mathbf B_{d}}(t)$ as in property \pref{p:5} and $[\cdot]_i$ extracts the i-th column from the argument. Since the inner integrals can be computed numerically, we have now a quadratic formulation of the cost and we can use it for the various constructions from \secref{sec:flat}, see also \cite{de2009flatness} for a similar treatment of cost computations.

For illustration purposes we consider $9$ hyperplanes:
\ben
H=\bbm 
   -0.5931   & \hphantom{-}0.8051\\
    \hphantom{-}0.1814   & \hphantom{-}0.9834\\
   -0.0044   & \hphantom{-}1.0000\\
   -0.1323   & \hphantom{-}0.9912\\
   -0.7011   &-0.7131\\
    \hphantom{-}0.8152   &-0.5792\\
    \hphantom{-}0.4352   & \hphantom{-}0.9003\\
    \hphantom{-}1.0000   &-0.0075\\
   -0.5961   &-0.8029   
\ebm, \quad h=\bbm 
    4.2239\\
    0.1719\\
    0.9975\\
    0.2728\\
    3.6785\\
    0.0317\\
    1.6598\\
    4.5790\\
    1.0280
\ebm
\een
which lead to a hyperplane arrangement \eqref{eq:hyparr} where the interdicted tuples \eqref{eq:sigma_int}
\begin{multline*}
\Sigma^\bullet=\{(+++--+++-),\\ (+-+-+++++), (+-+++--++)\}
\end{multline*}
correspond to three obstacles. Further, we take three way-points (initial, intermediary and final) though which the trajectory has to pass at predefined times:
\ben
\mathbb W=\left\{\bbm -9\\-0.5\ebm, \bbm 0\\1.5\ebm, \bbm 6\\0\ebm \right\} , \: \mathbb T_{\mathbb W}=\{0, 5, 10\}.
\een
\begin{figure*}[!ht]
\centering
\subfloat[various collision avoidance methods]{\label{fig:avoid_a}\includegraphics[width=.95\textwidth]{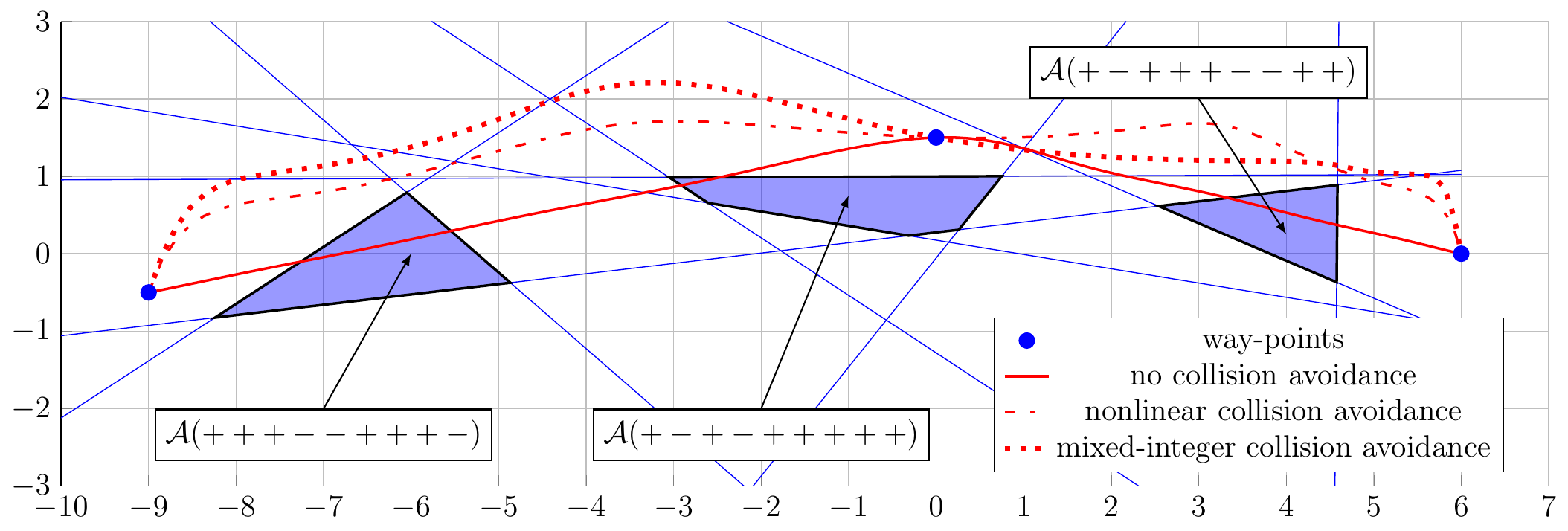}}\\
\subfloat[detail of the nonlinear case]{\label{fig:avoid_b}\includegraphics[width=.45\textwidth]{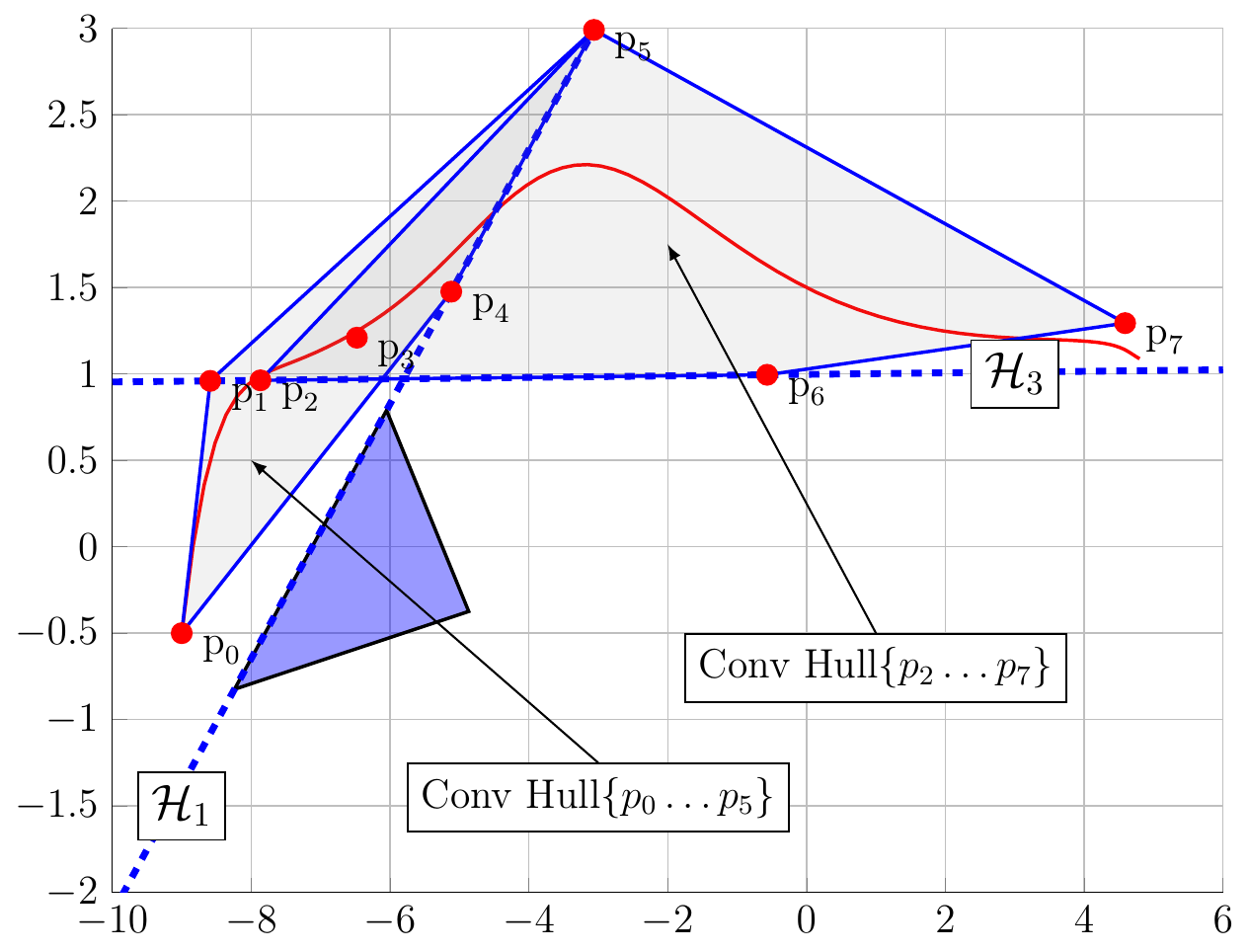}}
\subfloat[detail of the mixed-integer case]{\label{fig:avoid_c}\includegraphics[width=.45\textwidth]{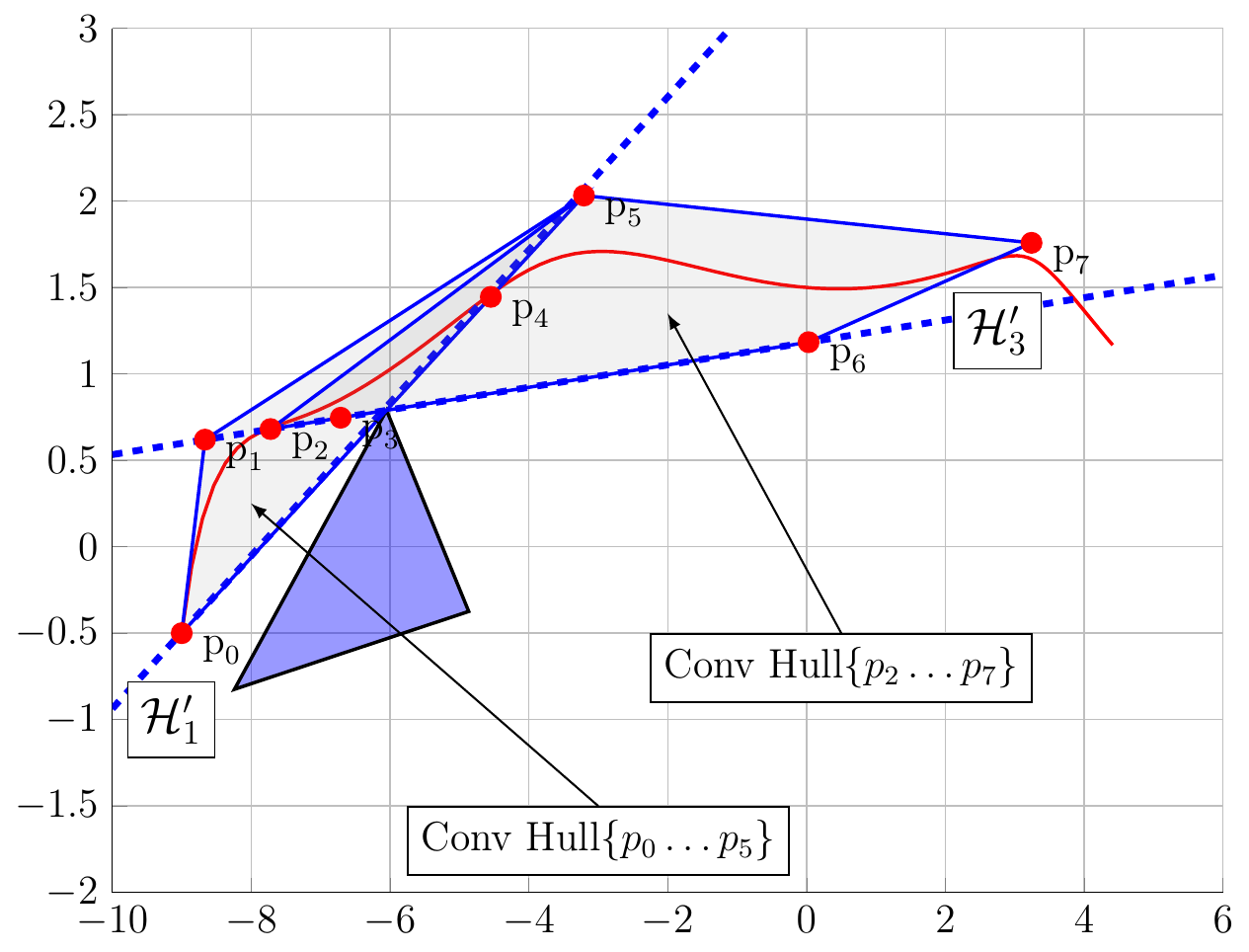}}
\caption{Flat trajectories in a multi-obstacle environment.}
\label{fig:avoid}
\end{figure*}

We compute a flat trajectory (parametrized by $n=12$ and $d=6$) which passes through the given way-points, minimizes the total path-length and respects one of the following scenarios: i) without any collision avoidance restriction; ii) with collision avoidance as in \propref{prop:avoid}; iii) with collision avoidance as in \cororef{cor:avoid}. The obstacles, their support hyperplanes and the resulting trajectories are depicted in \figref{fig:avoid_a}.

Scenarios ii) and iii) both accomplish the task of avoiding the obstacles with comparable computation times and path lengths. Figures \ref{fig:avoid_b} and \ref{fig:avoid_c} show details of the collision avoidance. In both cases the 1st and 3rd convex regions are considered ($\chull\{p_0\dots p_5\}$ and $\chull\{p_2\dots p_7\}$) together with their separating hyperplanes. In \figref{fig:avoid_b} these hyperplanes are the 1st and respectively 3rd support hyperplane whereas in \figref{fig:avoid_c} the separating hyperplanes are the result of the optimization problem and are $\mathcal H_1'=\{\bbm    -0.1273  &  0.2913\ebm^\top x=1 1\}$, $\mathcal H_3'=\{\bbm      -0.0549    0.8461\ebm^\top x=1 1\}$.

As mentioned earlier, the result of the optimization problem (computation time, total length of the trajectory) depends heavily on the number of control points $n+1$ and degree $d$. We illustrate these evolutions in \tabref{tab:evol}. Several remarks are in order. First, it seems that after an initial decrease in the path length the future reductions are negligible and at a significant computation time for the mixed-integer method. Next, and somewhat surprising, is that the non-linear method is extremely sensitive to parameter variations (number of control points, degree, positioning and number of the way-points, etc) such that the results obtained are not trustworthy. The one advantage of the latter over the former is that it may provide a feasible solution for small values of $n$.

\begin{table}
\begin{tabular}{|l|ccccc|}
\hline
  n& 10  & 15 &  20  & 25 &  30\\\hline   
	$t_{MI}$& 0.0985 &   6.231 &  17.106  & 15.034  & 21.522\\
  $\ell_{MI}$&       *&   16.877 &  16.307 &  16.202 &  16.536\\
	$t_{NL}$&37.426 &   0.6412 &   1.055 &   0.6321   & *\\
  $\ell_{NL}$ & 203.51   & 33.469   & 202.076  &  31.159  &  *\\\hline   
\end{tabular}
\caption{Evolution of trajectories characteristics for degree $d=4$ in the non-linear and mixed-integer formulations.}
\label{tab:evol}
\end{table}
 
The collision avoidance between two agents is similar and not depicted here. All the numerical simulations have been done using Yalmip \cite{YALMIP} and MPT Toolboxes \cite{mpt} in Matlab 2013a. The nonlinear solver used was the IPOPT solver \cite{ipopt2006}.

\section{Conclusions}
\label{sec:con}

This paper considers collision avoidance in a multi-agent multi-obstacle framework. Using differential flatness for trajectory generation and B-splines for the flat output parametrization we show that the restriction can be validated at all times. Exact and sub-optimal constructions are provided. Future work may consist in analysis of the feasibility of the problem, relaxation of way-point restrictions (similar with work done in \cite{med2015}), etc.

\section*{Acknowledgments}
This work was supported by a grant of the Romanian National Authority for Scientific Research and Innovation, CNCS – UEFISCDI, project number PN-II-RU-TE-2014-4-2713.

\end{document}